\documentclass[letterpaper,10pt,conference]{ieeeconf}

\usepackage{graphics} 
\usepackage{tikz}
\usetikzlibrary{positioning}
\usetikzlibrary{arrows}
\usetikzlibrary{automata}
\usepackage{cite}

\usepackage{epstopdf}
\usepackage{epsfig} 
\usepackage{subcaption}
\usepackage{psfrag}
\usepackage{caption}
\usepackage{lipsum}
\usepackage{amsmath} 
\usepackage{amssymb}  
\usepackage{mathbbol}

\usepackage{wrapfig}
\usepackage{mathrsfs}
\usepackage{url}
\usepackage{amsmath,amsfonts,amssymb,color}

\usetikzlibrary{shapes,arrows,calc}
\usepackage{marvosym}

\newtheorem{theorem}{Theorem}
\newtheorem{corollary}{Corollary}

\newtheorem{remark}{Remark}
\newtheorem{proposition}{Proposition}

\newtheorem{definition}{Definition}
\newtheorem{example}{Example}

\newcommand{\ignore}[1]{}

\newcommand{\oprocendsymbol}{\hbox{$\square$}}
\newcommand{\oprocend}

\DeclareMathAlphabet{\mathcal}{OMS}{cmsy}{m}{n}      

\newcommand{\FF}{\mathcal{F}}

\newcommand{\Eb}{\mathbb{E}}
\newcommand{\Pb}{\mathbb{P}}

\renewcommand{\AA}{\mathcal{A}}

\newcommand{\VV}{\mathcal{V}}

\newcommand{\St}{\mathtt{S}}
\newcommand{\SXt}{\mathtt{SX}}
\newcommand{\SYt}{\mathtt{SY}}
\newcommand{\Xt}{\mathtt{X}}
\newcommand{\Yt}{\mathtt{Y}}

\allowdisplaybreaks

\IEEEoverridecommandlockouts

\title{A Generalized SIS Epidemic Model on Temporal Networks with \\ Asymptomatic Carriers and Comments on Decay Ratio}

\author{Ashish~R.~Hota~and~Kavish~Gupta\thanks{The authors are with the Department of Electrical Engineering, Indian Institute of Technology (IIT) Kharagpur, India. E-mail: ahota@ee.iitkgp.ac.in, kavishgupta1999@gmail.com.}}%

\begin{document}

\maketitle

\begin{abstract}
We study the class of SIS epidemics on temporal networks and propose a new activity-driven and adaptive epidemic model that captures the impact of asymptomatic and infectious individuals in the network. In the proposed model, referred to as the A-SIYS epidemic, each node can be in three possible states: susceptible, infected without symptoms or asymptomatic and infected with symptoms or symptomatic. Both asymptomatic and symptomatic individuals are infectious. We show that the proposed A-SIYS epidemic captures several well-established epidemic models as special cases and obtain sufficient conditions under which the disease gets eradicated by resorting to mean-field approximations. 

In addition, we highlight a potential inaccuracy in the derivation of the upper bound on the decay ratio in the activity-driven adaptive SIS (A-SIS) model in \cite{ogura2019optimal} and present a more general version of their result. We numerically illustrate the evolution of the fraction of infected nodes in the A-SIS epidemic model and show that the bound in \cite{ogura2019optimal} often fails to capture the behavior of the epidemic in contrast with our results.
\end{abstract}

\section{Introduction}
\label{section:introduction}

The susceptible-infected-susceptible (SIS) epidemic is one of the most well-studied class of spreading processes on networks \cite{hethcote2000mathematics,pastor2015epidemic}. Early work on SIS epidemics focused on analyzing both deterministic \cite{mei2017dynamics} and stochastic \cite{pastor2015epidemic,van2009virus} dynamic evolution of the epidemic states; often resorting to mean-field approximations for analytical tractability. Most of the existing work has analyzed the epidemic dynamics on static networks; both deterministic as well as large-scale complex networks with a structured population \cite{pastor2015epidemic,van2009virus}.

However, the contact pattern in the human population is dynamic and time-varying. Furthermore, during the prevalence of an infectious disease, individuals often take precautions and reduce their social activities to protect themselves and others from becoming infected. Thus, the characteristics of the network or contact pattern evolves in a time-scale that is comparable to the evolution of the epidemic. Consequently, several recent works have analyzed epidemic processes on temporal or dynamical networks \cite{ogura2016stability,pare2017epidemic,ogura2017optimal,masuda2017temporal,enright2018epidemics,leitch2019toward}. 

In this work, we consider the class of SIS epidemics within the activity-driven network paradigm which is a relatively simple yet expressive paradigm for analyzing the evolution of epidemics and contact pattern in a comparable time-scale \cite{perra2012activity,zino2017analytical}. Our work is motivated by and builds upon the recent works \cite{ogura2019optimal,zino2020assessing} that study SIS epidemics and their close variants on activity-driven networks. Specifically, \cite{ogura2019optimal} defines the discrete-time {\it activity-driven adaptive-SIS model} (activity-driven A-SIS model), derives an analytical upper bound on the {\it decay ratio} of the infection probabilities of the nodes and proposes tractable optimization problems for optimal containment of the epidemic by minimizing the bound on the decay ratio. Similarly, in \cite{zino2020assessing}, the authors study a continuous-time SAIS epidemic with an additional state that captures individuals who are alert and protect themselves from the epidemic. The authors derive conditions for epidemic persistence and investigate optimal policies to mitigate the epidemic by reducing activation probabilities of infected nodes and prompting self-protective behavior. 

Our work is motivated by infectious diseases where a subset of infected individuals do not develop symptoms despite being infectious, i.e., they act as asymptomatic carriers; examples include COVID-19 \cite{hu2020clinical,park2020time} and Ebola \cite{bellan2014ebola}. Such individuals are often not aware of being infected and do not reduce their activity and contact patterns. As a result, such diseases are often challenging to contain. However, the above characteristic is not captured by the classical SIS epidemic model and its well-established variants that have additional states such as alert, exposed, etc. While recent papers \cite{pang2020public,chisholm2018implications} have highlighted the impacts of such asymptomatic carriers on the evolution and control of epidemics, rigorous and quantitative analysis of the above characteristic are few in the existing work on epidemics (on temporal networks).

In this paper, we propose a new activity-driven and adaptive generalized SIS epidemic model, referred to as the A-SIYS epidemic, where we treat asymptomatic and symptomatic individuals as distinct infection states (see Section \ref{section:a-siys} for a formal definition and discussion). Our model captures several well-established epidemic models as special cases. Furthermore, in our setting, each node potentially chooses a different number of nodes to connect to; this is in contrast with the homogeneity assumption in \cite{zino2020assessing,ogura2019optimal}. We derive a linearized dynamics that upper bounds the Markovian evolution of the epidemic states via a mean-field approximation and present sufficient conditions under which the epidemic gets eradicated. Our results and proof techniques are inspired by the analysis in \cite{ogura2019optimal}.

As a second contribution, we highlight a potential inaccuracy in the derivation of the upper bound on the decay ratio of the A-SIS epidemic model in \cite{ogura2019optimal} and obtain a counterpart of their result for a more general setting where nodes choose different numbers of other nodes to connect to (Section \ref{section:asis_error}). We then simulate the epidemic models for various parameter settings and show that the bound obtained in \cite{ogura2019optimal} does not always capture the behavior of the epidemic, in contrast with our results (Section \ref{section:simulation}). We conclude with a discussion on open problems and avenues for future research (Section \ref{sec:discussion}). 

\section{Activity-Driven Adaptive SIYS Epidemic}
\label{section:a-siys}

%%%% Figure %%%%%

\begin{figure}[tb]
\centering
\begin{tikzpicture}[font=\sffamily]

% Setup the style for the states
\tikzset{node style/.style={state, minimum width=1cm, line width=0.2mm, fill=green!60!white}}
\tikzset{node style1/.style={state, minimum width=1cm, line width=0.2mm, fill=red!50!white}}

        % Draw the states
\node[node style] at (0, 0)     (St)     {$\St$};
\node[node style1] at (5, 0)     (At)     {$\Xt$};
\node[node style1] at (2.5, -2) (Yt) {$\Yt$};

        % Connect the states with arrows
        \draw[every loop,
              auto=right,
              line width=0.4mm,
              >=latex,
              draw=black,
              fill=black]
            %(St)     edge[bend right=20]            node {0.1} (Yt)
            (St)     edge[bend right=0, auto=left] node[below] {$\beta_x (\Xt), \beta_y (\Yt)$} (At)
            (At)     edge[bend right=30]            node {$\delta_x$} (St)
            (At) edge node[below] {$\nu$} (Yt)
            %(Yt) edge[bend right=20]            node {0.2} (At)
            (Yt) edge node[below] {$\delta_y$} (St);
    \end{tikzpicture}
\caption{\footnotesize Probabilistic evolution of states in the A-SIYS epidemic model. Self-loops are omitted for better clarity. See Definition \ref{def:siys} for the formal definition. Red indicates that both $\Xt$ and $\Yt$ are infected states.}
\label{fig:siys_tran}
\end{figure}
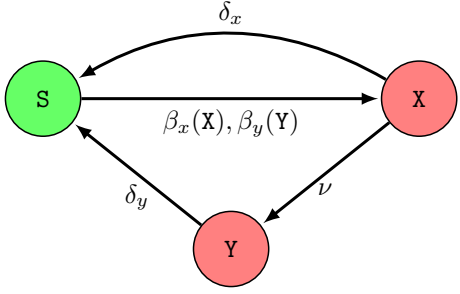

In this section, we formally define the activity-driven adaptive SIYS (A-SIYS) epidemic model. Let $\VV = \{v_1, v_2, \ldots, v_n\}$ denote the set of $n$ nodes. Each node remains in one of the three possible states: susceptible ($\St$), infected without symptoms or asymptomatic ($\Xt$) and infected with symptoms or symptomatic ($\Yt$). Both asymptomatic and symptomatic individuals are infectious, which captures the characteristics of certain epidemics such as COVID-19. 

The states evolve in discrete-time. If at time $t \in \{0,1,\ldots\}$, node $v_i$ is susceptible (respectively, asymptomatic and symptomatic), we denote this by $v_i(t) \in \St$ (respectively, $v_i(t) \in \Xt$ and $v_i(t) \in \Yt$). Given a network or contact pattern, the probabilistic state evolution is defined below.

\smallskip

\begin{definition}\label{def:siys}
Let $\beta_x, \beta_y, \delta_x, \delta_y, \nu \in [0, 1]$ be constants pertaining to infection, recovery and transition rates. The state of each node $v_i$ evolves as follows.
\begin{enumerate}
\item If $v_i(t) \in \St$, then $v_i(t+1) \in \Xt$ with probability $\beta_x$ for each asymptomatic neighbor and with probability $\beta_y$ for each symptomatic neighbor independently of other neighbors. 
\item If $v_i(t) \in \Xt$, then $v_i(t+1) \in \St$ with probability $\delta_x$ and $v_i(t+1) \in \Yt$ with probability $\nu(1-\delta_x)$. 
\item If $v_i(t) \in \Yt$, then $v_i(t+1) \in \St$ with probability $\delta_y$. 
\end{enumerate}
The state remains unchanged otherwise. \hfill \oprocendsymbol
\end{definition}

\smallskip

The possible transitions of the states are illustrated in Figure \ref{fig:siys_tran}. Thus, in our model, both asymptomatic and symptomatic nodes can potentially infect a susceptible node, albeit with different probabilities ($\beta_x$ and $\beta_y$, respectively). Upon being infected, a susceptible node becomes asymptomatic. From there on, it can either get cured and become susceptible with probability $\delta_x$, and if not, it transitions to the symptomatic state with probability $\nu$. Thus $\nu^{-1}$ captures the delay in onset of symptoms. The curing rate for symptomatic nodes is $\delta_y$. Thus, in our model, a node can get infected and cured without ever exhibiting symptoms. 

With the above definition in place, we now formally define the activity-driven and state-dependent evolution of the network or contact pattern and the epidemic states of individual nodes. As discussed above, our model builds upon the formulation in \cite{ogura2019optimal} for the A-SIS epidemic model.

\smallskip

\begin{definition}\label{def:act_siys}
For each node $v_i \in \VV$, let $a_i, \chi_i, \pi_i \in (0,1]$ be constants referred to as the activity rate, adaptation factor and acceptance rate of $v_i$, respectively. Let $m_i \geq 1$ be the number of nodes $v_i$ attempts to connect to upon activation. Let $\beta_x, \beta_y, \delta_x, \delta_y, \nu \in [0, 1]$ be constants pertaining to infection and recovery rates. The A-SIYS model is defined by the following procedures:
\begin{enumerate}
\item At the initial time $t = 0$, each node is in one of the three possible states.
\item At each time $t = 1, 2, \ldots$, each node $v_i$ randomly becomes activated independently of other nodes with the following probability:
\begin{equation*}
\Pb(\text{$v_i$ becomes activated}) \!= \!
\begin{cases}
a_i, \!\!\!& \text{ if $v_i(t) \in \St \cup \Xt$},
\\ \chi_i a_i, \!\!\!& \text{ if } v_i(t) \in \Yt.
\end{cases}
\end{equation*}
\item Node $v_i$, upon activation, randomly and uniformly chooses $m_i$ other nodes independently of other activated nodes. If $v_j$ is chosen by $v_i$, an edge $(v_i,v_j)$ is created with the following probability:
\begin{equation*}
\Pb(\text{edge } (v_i,v_j) \text{ is created}) = 
\begin{cases}
1, \!\!\!& \text{if $v_j(t) \in \St \cup \Xt$},
\\ \pi_j, \!\!\!& \text{if $v_j(t) \in \Yt$}.
\end{cases}
\end{equation*}
These edges are discarded at time $t + 1$. 
\item Once the edges are formed following the above procedure, the states of the nodes get updated following Definition \ref{def:siys}. 
\item Steps 2-4 are repeated for each time $t \geq 1$. \hfill \oprocendsymbol
\end{enumerate} 
\end{definition}

\smallskip

Thus, for node activation and link formation, susceptible and asymptomatic (who are not aware of being infected) nodes behave in an identical manner. When $\chi_i \in (0,1)$, the probability of node $v_i$ getting activated when it is symptomatic is smaller than its activation probability when it is susceptible or asymptomatic. This is potentially due to sickness or reduction of activities by $v_i$ so as to not infect others when it learns that it is infected. Similarly, a symptomatic node is less likely to accept an edge compared to a susceptible or asymptomatic node.

\begin{remark}\label{remark:asiys_special}
The A-SIYS epidemic defined above is quite general and captures the following models as special cases. 
\begin{enumerate}
\item $\nu = 0, v_i(0) \notin \Yt, \forall v_i \in \VV$: In this case, the nodes never enter the symptomatic state, and the activation and acceptance rates no longer state-dependent; the latter parameters are $a_i$ and $1$, respectively. Thus, the epidemic behaves as the classical SIS epidemic on an activity-driven network (but non-adaptive). 
\item $\beta_x = 0, \delta_x = 0$: Here, an asymptomatic node is not infectious and eventually becomes symptomatic before becoming susceptible. In this regime, our model is the activity-driven and adaptive analogue of the SEIS epidemic \cite{hethcote2000mathematics} with $\Xt$ being the ``exposed" state. \hfill \oprocendsymbol
\end{enumerate}
\end{remark}

\subsection{State Evolution and Mean-Field Approximation}

In order to analyze the evolution of the states in the A-SIYS epidemic model, we define random variables $S_i(t), X_i(t)$ and $Y_i(t)$ associated with node $v_i$ that take values in the set $\{0,1\}$. Specifically, we define $S_i(t) = 1$ if $v_i(t) \in \St$, $X_i(t) = 1$ if $v_i(t) \in \Xt$ and $Y_i(t) = 1$ if $v_i(t) \in \Yt$. Since a node can only be in one of three possible states, we have $S_i(t) + X_i(t) + Y_i(t) = 1$. Similarly, we define a $\{0,1\}$-valued random variable $A_{ij}(t)$ which takes value $1$ if the edge $(v_i,v_j)$ exists at time $t$. We also denote by $N_x$ a Bernoulli random variable that takes value $1$ with probability $x \in [0,1]$. The state transition of $v_i$ under the A-SIYS epidemic model can now be formally stated as:
\begin{subequations}\label{eq:asiys_markov}
\begin{align}
& S_i(t+1) = S_i(t) \underset{j \neq i}{\Pi} \left[ 1 - A_{ij}(t) (X_j(t)N_{\beta_x} + Y_j(t)N_{\beta_y}) \right] \nonumber
\\ & \qquad \qquad + N_{\delta_x} X_i(t) + N_{\delta_y} Y_i(t),
\\ & X_i(t+1) = (1-N_{\delta_x})(1-N_{\nu}) X_i(t) + \nonumber
\\ & \quad \!S_i(t) \!\big[1 \!- \underset{j \neq i}{\Pi} \left[ 1 \!- A_{ij}(t) (X_j(t)N_{\beta_x} \!+ Y_j(t)N_{\beta_y}) \!\right]\big], \label{eq:asiys_markov_x}
\\ & Y_i(t+1) = (1-N_{\delta_x})N_{\nu} X_i(t) + (1-N_{\delta_y})Y_i(t). \label{eq:asiys_markov_y}
\end{align}
\end{subequations}
It is easy to see that $S_i(t+1) + X_i(t+1) + Y_i(t+1) = S_i(t) + X_i(t) + Y_i(t) = 1$. We denote the probability of node $v_i$ being susceptible at time $t$ by $s_i(t)$, i.e., $s_i(t) := \Pb(v_i(t) \in \St) = \Eb[S_i(t)]$. The quantities $x_i(t)$ and $y_i(t)$ are defined in an analogous manner. 

Note that the infection states follow a Markov process with a $3^n \times 3^n$ transition probability matrix with the state $S_i(t)=1$ for all $v_i \in \VV$ (i.e., the disease-free state) being the only absorbing state. While analyzing the behavior of this model is computationally intractable, we rely on a mean-field approximation and upper bound the evolution of the infection probability (both asymptomatic and symptomatic) via a linear dynamics. We then derive sufficient conditions under which the epidemic decays to the disease-free state. We start with the following result.

\begin{theorem}\label{thm:asiys_main}
Consider the activity-driven adaptive SIYS (A-SIYS) epidemic model defined in Definition \ref{def:act_siys}.  Let $\bar{m}_i = m_i/(n-1)$, and for all $i, j$, define the constants
\begin{align}
\beta^{ij}_x & := \beta_x [1-(1-a_i \bar{m}_i)(1 - a_j \bar{m}_j)] \label{eq:betaxij},
\\ \beta^{ij}_y & := \beta_y [1-(1-a_i \bar{m}_i \pi_j)(1 - \chi_j a_j \bar{m}_j)]. \label{eq:betayij}
\end{align}
Then,
\begin{subequations}\label{eq:mfa_linear_dynamics}
\begin{align}
& x_i(t+1) \!\leq \delta^c_x (1-\nu) x_i(t) \!+ \!\sum_{j \neq i} \!\left[\beta^{ij}_x x_j(t) \!+\! \beta^{ij}_y y_j(t) \right], 
\\ & y_i(t+1) = \delta^c_x \nu x_i(t) + (1-\delta_y)y_i(t), 
\end{align}
\end{subequations}
for all nodes $v_i$ and $t \geq 0$ with $\delta^c_x = 1-\delta_x$. \hfill \oprocendsymbol
\end{theorem}

\begin{proof}
We compute expectation on both sides of \eqref{eq:asiys_markov_x} and \eqref{eq:asiys_markov_y} and obtain
\begin{subequations}\label{eq:asiys_mfa}
\begin{align}
& x_i(t+1) = (1-\delta_x)(1-\nu) x_i(t) + \nonumber
\\ & \quad \Eb\Big[S_i(t) \big[1 \!- \!\underset{j \neq i}{\Pi} \!\left[1 \!- A_{ij}(t) (X_j(t)N_{\beta_x} \!+ Y_j(t)N_{\beta_y}) \!\right]\big]\Big], \label{eq:asiys_mfa_x}
\\ & y_i(t+1) = (1-\delta_x)\nu x_i(t) + (1-\delta_y)y_i(t),
\end{align}
\end{subequations}
and $s_i(t) = 1 - x_i(t) - y_i(t)$. For the product term in the R.H.S. of \eqref{eq:asiys_mfa_x}, the Weierstrass product inequality yields
\begin{align*}
& 1 - \underset{j \neq i}{\Pi} \left[ 1 - A_{ij}(t) (X_j(t)N_{\beta_x} + Y_j(t)N_{\beta_y}) \right] 
\\ & \qquad \leq \sum_{j\neq i} A_{ij}(t) (X_j(t)N_{\beta_x} + Y_j(t)N_{\beta_y}). 
\end{align*}
Consequently, we have
\begin{align}
& \Eb\Big[S_i(t) \big[1 - \underset{j \neq i}{\Pi} \left[ 1 - A_{ij}(t) (X_j(t)N_{\beta_x} + Y_j(t)N_{\beta_y}) \right]\big]\Big] \nonumber
\\ & \leq \!\sum_{j \neq i} \!\beta_x \Eb[A_{ij}(t) S_i(t) X_j(t)] \!+\! \beta_y \Eb[A_{ij}(t) S_i(t) Y_j(t)]. \label{eq:ineq_main}
\end{align}

We now focus on evaluating the expectation terms in the above equation. Recall that $A_{ij}(t)$ is a random variable that indicates the presence of the edge $(v_i,v_j)$ at time $t$ and is governed by the states of nodes $v_i$ and $v_j$ according to Definition \ref{def:act_siys}. In order to bound the expectation terms, we introduce the following notation for events of interest:
\begin{align}
& \SXt^t_{ij} = ``v_i(t) \in \St \text{ and } v_j(t) \in \Xt,"
\\ & \SYt^t_{ij} = ``v_i(t) \in \St \text{ and } v_j(t) \in \Yt,"
\\ & \Gamma^t_{i\to j} = \text{``$v_i$ is activated, chooses } v_j \text{ as neighbor at t."} \label{eq:def_Gammatij}
\end{align}
With the above notation in place, we have
\begin{subequations}
\label{eq:sub_cond_exp}
\begin{align}
\Eb[A_{ij}(t) S_i(t) X_j(t)] & = \Pb(A_{ij}(t)=1|\SXt^t_{ij})\Pb(\SXt^t_{ij}), \label{eq:sub_cond_exp1}
\\ \Eb[A_{ij}(t) S_i(t) Y_j(t)] & = \Pb(A_{ij}(t)=1|\SYt^t_{ij})\Pb(\SYt^t_{ij}). \label{eq:sub_cond_exp2}
\end{align}
\label{eq:sub_cond_exp}
\end{subequations}
We now focus on the first equation above and note that
\begin{align*}
& \Pb(A_{ij}(t)=1|\SXt^t_{ij}) = \Pb(\Gamma^t_{i \to j} | \SXt^t_{ij}) + \Pb(\Gamma^t_{j \to i} | \SXt^t_{ij}) 
\\ & \qquad \quad - \Pb(\Gamma^t_{i \to j} | \SXt^t_{ij}) \Pb(\Gamma^t_{j \to i} | \SXt^t_{ij})
\\ & \qquad = 1-[1-\Pb(\Gamma^t_{i \to j} | \SXt^t_{ij})][1-\Pb(\Gamma^t_{j \to i} | \SXt^t_{ij})].
\end{align*}
Since the event $\SXt^t_{ij}$ states that $v_i$ is susceptible and $v_j$ is infected without symptoms, according to Definition \ref{def:act_siys}, the adaptation and acceptance of node $v_j$ is same as the case when it is susceptible. Therefore, we have
\begin{align*}
\Pb(\Gamma^t_{i \to j} | \SXt^t_{ij}) & = a_i \bar{m}_i, \qquad \Pb(\Gamma^t_{j \to i} | \SXt^t_{ij}) = a_j \bar{m}_j.
\end{align*}
Similarly for events conditioned on $\SYt^t_{ij}$, we have 
\begin{align*}
\Pb(A_{ij}(t)=1|\SYt^t_{ij}) & = 1-[1-\Pb(\Gamma^t_{i \to j} | \SYt^t_{ij})] \times
\\ & \qquad [1-\Pb(\Gamma^t_{j \to i} | \SYt^t_{ij})].
\end{align*}
The event $\SYt^t_{ij}$ corresponds to $v_i$ being susceptible and $v_j$ being infected with symptoms. Therefore, the adaptation and acceptance of $v_j$ depend on the parameters $\chi_j$ and $\pi_j$, respectively. Therefore, following Definition \ref{def:act_siys}, we have 
\begin{align*}
\Pb(\Gamma^t_{i \to j} | \SYt^t_{ij}) & = a_i \bar{m}_i \pi_j, \qquad \Pb(\Gamma^t_{j \to i} | \SYt^t_{ij}) = \chi_j a_j \bar{m}_j.
\end{align*}
Finally, we note that
\begin{align*}
\Pb(\SXt^t_{ij}) \leq x_j(t), \qquad \Pb(\SYt^t_{ij}) \leq y_j(t),
\end{align*}
since the event $v_j(t) \in \Xt$ (respectively, $v_j(t) \in \Yt$) subsumes the event $\SXt^t_{ij}$ (respectively, $\SYt^t_{ij}$). 

Substituting the above bounds and the expressions for the conditional probabilities obtained in \eqref{eq:sub_cond_exp}, we obtain
\begin{align*}
& \Eb[A_{ij}(t) S_i(t) X_j(t)] \leq [1-(1-a_i \bar{m}_i)(1 - a_j \bar{m}_j)] x_j(t), 
\\ & \Eb[A_{ij}(t) S_i(t) Y_j(t)] \leq \![1\!-(1\!-a_i \bar{m}_i \pi_j)(1 \!- \chi_j a_j \bar{m}_j)] y_j(t).
\end{align*}
The result now follows upon substituting the above expressions in \eqref{eq:ineq_main} and the definition of $\beta^{ij}_x$ and $\beta^{ij}_y$.
\end{proof}

The above result shows that the evolution of the probability of a node being asymptomatic and symptomatic is upper bounded by a linear dynamics as stated in \eqref{eq:mfa_linear_dynamics}. The linearized dynamics can be stated in a compact manner as follows. Let $z(t) := [x(t)^\top \quad y(t)^\top]^\top \in [0,1]^{2n}$ be the vector of probabilities corresponding to the infected states. From the above theorem, we have 
\begin{equation}
z(t+1) \leq \left[
\begin{array}{c|c}
\AA_{xx} & \AA_{xy} \\
\hline
\AA_{yx} & \AA_{yy}
\end{array}
\right] z(t) =: \AA z(t), 
\end{equation}
where each sub-matrix has dimension $n \times n$. Specifically, $\AA_{xx}$ has diagonal entries $(1-\delta_x)(1-\nu)$ and $(i,j)$-th entry as $\beta^{ij}_x$ for $j \neq i$, $\AA_{xy}$ has diagonal entries $0$ and $\beta^{ij}_y$ as the $(i,j)$-th entry with $j \neq i$, $\AA_{yx} := \mathtt{diag}((1-\delta_x) \nu)$, and $\AA_{yy} := \mathtt{diag}(1-\delta_y)$. Consequently, we obtain a sufficient condition, stated below, under which the disease is eradicated.

\begin{theorem}\label{thm:asiys_sufficient}
The decay ratio of the epidemic
\begin{align*}
\alpha & := \inf\{\gamma: \text{ there exists } C > 0 \text{ such that } ||z(t)|| \leq C\gamma^t \\ & \quad \qquad \text{ for all } t \geq 0 \text{ and } z(0)\}
\end{align*} 
is upper bounded as $\alpha \leq \rho(\AA)$ where $\rho(\AA)$ is the spectral radius of $\AA$. In particular, if $\rho(\AA) < 1$, $\underset{t\to\infty}{\lim} ||z(t)|| = 0$.  \hfill \oprocendsymbol
\end{theorem}

The proof follows from the above discussion and standard arguments and is omitted in the interest of space. Note further that $\AA$ is a non-negative irreducible (since each node can potentially choose any other node to connect to) matrix. Thus, $\rho(\AA)$ corresponds to its largest eigenvalue which is real and positive following the Perron-Forbenius theorem \cite{horn2012matrix}. 

We now state the following corollaries of the above results that correspond to certain special cases of our model. 

\begin{corollary}
Suppose $\nu = 0$ and $v_i(0) \notin \Yt, \forall v_i \in \VV$. Then, $y_i(t) = 0$ for all $t \geq 0$ and 
\begin{equation*}
x_i(t) \leq (1-\delta_x) x_i(t) + \sum_{j \neq i} \beta^{ij}_x x_j(t). \end{equation*}
Furthermore, $\alpha \leq \rho(\AA_{xx})$. \hfill \oprocendsymbol
\end{corollary} 

The above setting corresponds to the classical SIS epidemic on an activity-driven network discussed in Remark \ref{remark:asiys_special}. Note from the definition of $\beta^{ij}_x$ that the matrix $\AA_{xx}$ consists of a diagonal matrix and a matrix of rank $2$, and consequently, its spectral radius can be explicitly derived. 

When symptomatic individuals completely stop interacting with others, we have the following corollary. 

\begin{corollary}
Suppose $\chi_i = 0$ and $\pi_i = 0$ for all the nodes. Then, $\beta^{ij}_y = 0$ and we have
\begin{equation*}
z(t+1) \leq \left[
\begin{array}{c|c}
\AA_{xx} & \mathbf{0}_{n \times n} \\
\hline
\AA_{yx} & \AA_{yy}
\end{array}
\right] z(t), 
\end{equation*}
where $\mathbf{0}_{n \times n}$ has all entries equal to $0$. The bound on the decay ratio is given by $\alpha \leq \max(1-\delta_y, \rho(\AA_{xx}))$. \hfill \oprocendsymbol
\end{corollary} 

The above regime corresponds to situations where symptomatic individuals are kept in strict isolation. Our analysis shows that even when the decay ratio pertaining to interaction among nodes (corresponding to $\rho(\AA_{xx})$) is small, the epidemic eradication rate (dominated by $1-\delta_y$) can be slow if the recovery rate $\delta_y$ is sufficiently small. 

As discussed earlier, the proposed model and the above analysis is motivated by and builds upon the Activity-Driven Adaptive SIS (A-SIS) epidemic proposed in \cite{ogura2019optimal}. In the following section, we highlight a potential inaccuracy in the derivation of the upper bound on the decay ratio in \cite{ogura2019optimal}. 

\section{Activity-Driven A-SIS Epidemic Model and Potential Inaccuracy in the Analysis of \cite{ogura2019optimal}}\label{section:asis_error}

The activity-driven A-SIS epidemic is not a special case of the A-SIYS epidemic studied above, but is closely related. In the A-SIS epidemic defined in \cite{ogura2019optimal}, a node $v_i$ is either susceptible or infected. A susceptible node becomes infected with probability $\beta \in [0,1]$ when it comes in contact with an infected node (independently of other infected nodes) and an infected node recovers with probability $\delta \in [0,1]$. Thus, the state transition is a special case of Definition \ref{def:siys} when $\nu=0, \delta_x = \delta$ and $\beta_x = \beta$ and $\Xt$ denoting the infected state.

The models differ in the activity-driven adaptive network formation process. While the A-SIYS epidemic distinguishes between asymptomatic and symptomatic infections, the A-SIS epidemic does not. Specifically, upon infection, node $v_i$ adjusts its activation and acceptance probabilities with the factors $\chi_i$ and $\pi_i$, respectively as shown in points 2 and 3 in Definition \ref{def:act_siys}. Furthermore, each node upon activation chooses $m$ other nodes to connect to, i.e., $m_i = m$ for all $v_i \in \VV$. The rest of the steps are identical to those in Definition \ref{def:act_siys}. 

We follow the terminology in \cite{ogura2019optimal} and model the state of node $v_i$ at time $t$ as a random variable
\begin{equation}
x_i(t) := 
\begin{cases}
0 \qquad & \text{if $v_i$ is susceptible at time $t$},
\\ 1 \qquad & \text{if $v_i$ is infected at time $t$}.
\end{cases}
\end{equation}
Similarly, we define $p_i(t) := \Pb(v_i \text{ is infected at time $t$})$ and the vector of infection probabilities for all nodes as $p(t)$. The authors in \cite{ogura2019optimal} define {\it decay ratio} as follows. 

\smallskip

\begin{definition}[Definition 3.1 \cite{ogura2019optimal}]
We define the decay ratio of the activity-driven A-SIS model by
\begin{align*}
\alpha & = \inf\{\gamma: \text{ there exists } C > 0 \text{ such that } ||p(t)|| \leq C\gamma^t \\ & \quad \qquad \text{ for all } t \geq 0 \text{ and } x(0)\}. 
\end{align*} 
\end{definition}

\smallskip

The quantity $\alpha$ captures the persistence of infection among the nodes. In the A-SIS model, the states follow a Markov process and the actual decay ratio is the spectral radius of the $2^n \times 2^n$ transition probability matrix. The authors in \cite{ogura2019optimal} upper bound the evolution of $p_i(t)$ by a linear dynamics and obtain an explicit upper bound on the decay ratio by noting that it is the spectral radius of a matrix of rank $2$.  

However, we believe that the derivation of the linearized dynamics in Proposition 3.2 in \cite{ogura2019optimal} is inaccurate. We start our discussion by first stating Proposition 3.2 from \cite{ogura2019optimal}. 

\begin{proposition}[Proposition 3.2 \cite{ogura2019optimal}]
Let $\bar{m} = m/(n-1), \delta^c = 1-\delta$, and for all $i$, define the constants
$$ \phi_i = \bar{m} \chi_i a_i, \qquad \psi_i = \bar{m} \pi_i a_i. $$
Then,
\begin{equation}
p_i(t+1) \leq \delta^c p_i(t) + \beta \sum^n_{j=1} [1-(1-\psi_i)(1-\phi_j)]p_j(t)
\end{equation}
for all nodes $v_i$ and $t \geq 0$. \hfill \oprocendsymbol
\end{proposition}

\smallskip 

\noindent {\bf Potential inaccuracy in Proposition 3.2 \cite{ogura2019optimal}:} 

\smallskip 

\noindent The proof of Proposition 3.2 in \cite{ogura2019optimal} follows largely analogous steps as the proof of Theorem \ref{thm:asiys_main} above; the main distinction being the absence of terms related to $Y_j(t)$ in \cite{ogura2019optimal}. We believe that the evaluation of $\Pb(\Gamma^t_{i \to j}|\Xi^t_{i,j})$ in equation (3.19) in the proof in \cite{ogura2019optimal} is inaccurate. Note that
\begin{align*}
& \Xi^t_{i,j} := \text{``$v_i$ is susceptible and $v_j$ is infected at time $t$"}, 
\end{align*}
in equation (3.16) in \cite{ogura2019optimal} and $\Gamma^t_{i\to j}$ is as defined in \eqref{eq:def_Gammatij} above. Thus, $\Pb(\Gamma^t_{i \to j}|\Xi^t_{i,j})$ is the probability that the edge $(v_i,v_j)$ will be formed when initiated by the activated node $v_i$ when $v_i$ is susceptible and $v_j$ is infected. Thus, $\Pb(\Gamma^t_{i \to j}|\Xi^t_{i,j})$ is the product of $\Pb(v_i \text{ is activated while it is susceptible})$ and $\Pb((v_i,v_j) \text{ are neighbors when } v_j \text{ is infected})$. Following the definition of the A-SIS epidemic, we have 
\begin{align*}
\Pb(\Gamma^t_{i \to j}|\Xi^t_{i,j}) & = a_i \cdot \bar{m} \pi_j \neq \psi_i
\end{align*}
as $\psi_i = \bar{m} a_i \pi_i$. Since $v_j$ is infected at time due to the conditioning event $\Xi^t_{i,j}$, the probability of such an edge being formed is $\bar{m}\pi_j$ not $\bar{m}\pi_i$ as considered in \cite{ogura2019optimal}. {\it In other words, the probability that the edge $(v_i,v_j)$ will be formed when initiated by the activated node $v_i$ depends on the acceptance rate of the node $v_j$.} 

\smallskip 

\noindent {\bf Strengthening equation (3.21) in \cite{ogura2019optimal} when $i=j$:} 

\smallskip 

\sloppy 
\noindent The authors claim that the inequality 
$$ \Eb[(1-x_i(t))A_{ij}(t)x_j(t)] \leq [1-(1-\psi_i)(1-\phi_j)] p_j(t), $$
trivially holds when $i=j$. In fact, the event $\Xi^t_{i,i} = \text{``$v_i$ is susceptible and $v_i$ is infected at time $t$"}$ is empty and as a result $\Pb(\Xi^t_{i,i}) = 0$. Therefore, the bound can be strengthened by treating $ \Eb[(1-x_i(t))A_{ii}(t)x_i(t)] = 0$.

\smallskip

\noindent {\bf Implications:}

\smallskip

\noindent The above potential inaccuracy has significant implication on the bound derived in Theorem 3.3 in \cite{ogura2019optimal}. Specifically, the authors build upon Proposition 3.2 and show that the vector of infection probabilities evolves as
\begin{align*} 
p(t+1) & \!\leq \!\big[(1-\delta)\mathbf{I}_n \!+ \beta [\mathbf{1}_n \mathbf{1}_n^\top \!- (\mathbf{1}_n\!-\psi)(\mathbf{1}_n\!-\phi)^\top]\big] p(t) \\ & =: \FF p(t), 
\end{align*}
where $\mathbf{I}_n$ is the identity matrix and $\mathbf{1}_n$ is the vector of dimension $n$ with all entries being $1$. The authors then argue that the spectral radius of $\FF$, denoted $\rho(\FF)$, is an upper bound on the decay ratio with
$$ \rho(\FF) = 1-\delta + \beta \rho(\mathbf{1}_n \mathbf{1}_n^\top - (\mathbf{1}_n-\psi)(\mathbf{1}_n-\phi)^\top). $$
Since $\mathbf{1}_n \mathbf{1}_n^\top - (\mathbf{1}_n-\psi)(\mathbf{1}_n-\phi)^\top$ is a matrix of rank $2$, the authors could obtain an explicit expression on its spectral radius and consequently on the bound on the decay ratio. However, due to the potential inaccuracy highlighted above, the linear dynamics that bounds the evolution of the vector of infection probabilities is not necessarily a lower-ranked matrix. Furthermore, the contributions related to optimal resource allocation for containing the epidemic rely on the bound on the decay ratio and may no longer be applicable.

We now state the following theorem that addresses the above inaccuracy and generalizes the result in \cite{ogura2019optimal}. 

\begin{theorem}\label{thm:asis_main}
Consider a generalization of the activity-driven A-SIS model where node $v_i$, upon activation, chooses uniformly and randomly $m_i$ other nodes to connect to. Let $\bar{m}_i = m_i/(n-1)$, and for all $i, j$, define the constants
$$ \phi_i = \chi_i a_i \bar{m}_i, \qquad \psi_{ij} = a_i \bar{m}_i \pi_j. $$
Then,
\begin{equation}
p_i(t+1) \!\leq (1-\delta) p_i(t) \!+ \beta \sum_{j \neq i} [1\!-(1\!-\psi_{ij})(1\!-\phi_j)]p_j(t)
\end{equation}
for all nodes $v_i$ and $t \geq 0$. Furthermore, the decay ratio is upper bounded as $\alpha \leq \rho^* := \rho(\FF^*)$ or the spectral radius of the matrix $\FF^*$ with entries 
\begin{equation}
\FF^*_{ij} = 
\begin{cases}
& 1-\delta, \qquad \text{if $i=j$},
\\ & \beta \sum_{j \neq i} [1-(1-\psi_{ij})(1-\phi_j)], \qquad \text{if $i \neq j$}.
\end{cases}
\end{equation}
In particular, if $m_i = m$ for every node $v_i$, then the result holds with $\phi_i = \chi_i a_i \bar{m}$ and $\psi_{ij} = a_i \bar{m} \pi_j$ where $\bar{m} = m/(n-1)$. \hfill \oprocendsymbol
\end{theorem}

The proof largely mirrors the proof of Theorem \ref{thm:asiys_main} and the analysis in \cite{ogura2019optimal} with the above discussed aspects incorporated. We omit it in the interest of space.

\begin{remark}
The upper bound on the decay ratio as shown above is the spectral radius of an $n \times n$ matrix. While $\FF^*$ has a larger dimension than the case shown in \cite{ogura2019optimal}, it is still a considerable improvement over the $2^n \times 2^n$ matrix that characterizes the exact decay ratio. From the Perron-Frobenius theorem, the spectral radius also coincides with the largest eigenvalue of $\FF$. We further note that the results in \cite{ogura2019optimal} continues to hold if the acceptance rate is homogeneous across all the nodes, i.e., $\pi_i = \pi, \forall v_i \in \VV$. \hfill \oprocendsymbol
\end{remark}

\section{Simulation Results}
\label{section:simulation}

In this section, we illustrate the evolution of the epidemic states in both A-SIYS and A-SIS epidemic models. 

\subsection{Impact of transition rate $\nu$ in A-SIYS Epidemic} 

\begin{figure*}[tb]
	\centering
	\includegraphics[scale=1.1]{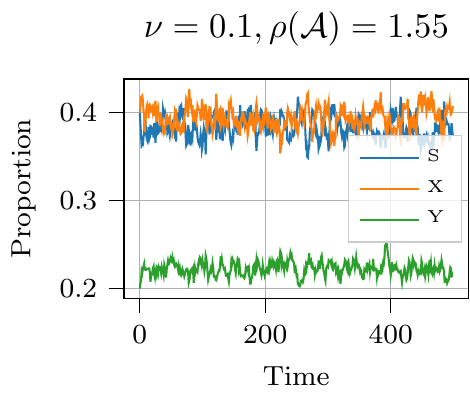}~~
	\includegraphics[scale=1.1]{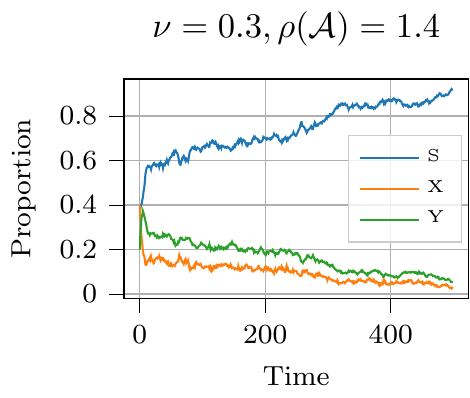}~~
	\includegraphics[scale=1.1]{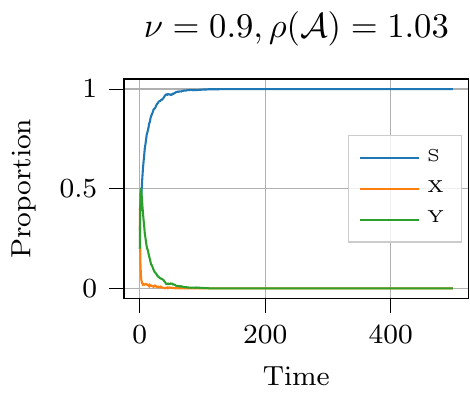}
	\caption{Evolution of proportion of susceptible (S), asymptomatic (X) and symptomatic (Y) nodes in the A-SIYS epidemic averaged across $50$ runs with parameters described in Example \ref{ex:asiys1}; $\nu$ denotes the rate at which an asymptomatic node becomes symptomatic and $\rho(\AA)$ is the bound on the decay ratio.}
	\label{fig:ex11}
\end{figure*}

\begin{figure*}[tb]
	\centering
	\includegraphics[scale=1]{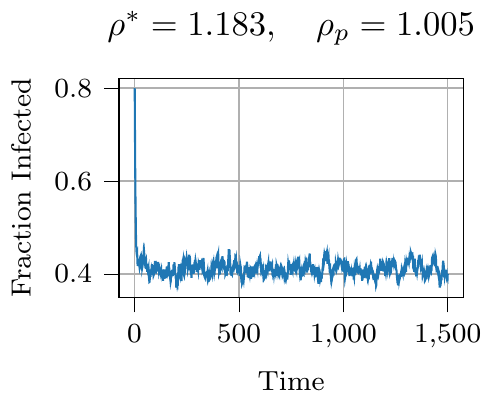}
	\includegraphics[scale=1]{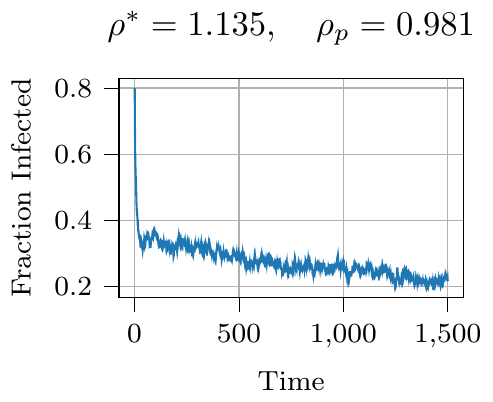}
	\includegraphics[scale=1]{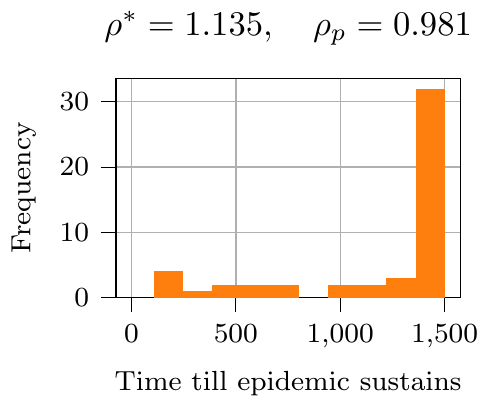} \\[2mm]
	\includegraphics[scale=1]{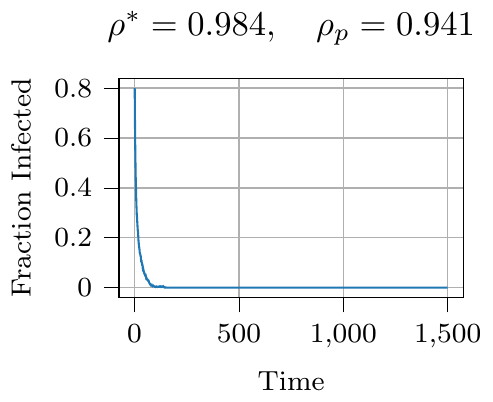}
	\includegraphics[scale=1]{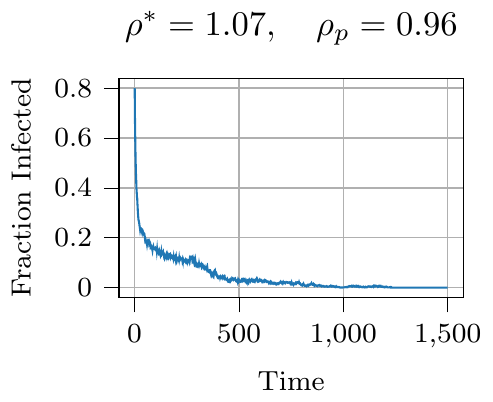}
	\includegraphics[scale=1]{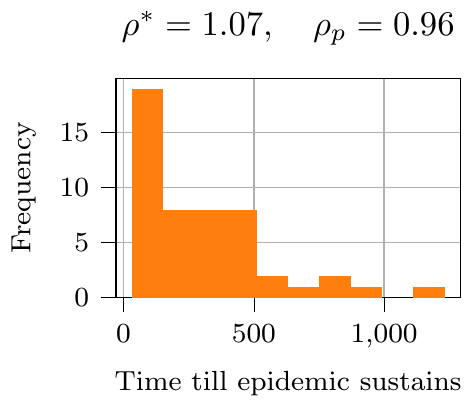}
	\caption{Evolution of fraction of infected nodes averaged across $50$ runs and histogram of time till the epidemic persists with parameters described in Example \ref{ex:1}.}
	\label{fig:ex1}
\end{figure*}

We first show the impact of asymptomatic carriers on the epidemic prevalence in the A-SIYS epidemic. 

\begin{example}\label{ex:asiys1}
We consider a set of $n = 50$ nodes and for each node set the rate of infection $\beta_x = \beta_y =  0.25$, rate of recovery $\delta_x = \delta_y =  0.15$, activity rate = $0.2$, adaptation factor = $0.05$, acceptance rate = $0.05$ and degree $m =  8$. We initialize with $20$ nodes being susceptible, $20$ nodes being asymptomatic and $10$ nodes being symptomatic and simulate the A-SIYS epidemic till $500$ time steps and $50$ independent runs. We show the evolution of the fraction of susceptible, asymptomatic and symptomatic nodes averaged over $50$ runs in Figure \ref{fig:ex11} for three different values of the transition rate $\nu = 0.1, 0.3, 0.9$. The upper bound on the decay ratio for these settings are $1.556, 1.4, 1.03$, respectively. 

Recall that $\nu$ captures the rate at which asymptomatic nodes become symptomatic. Given the above parameters, the adaptation and acceptance rates are negligible for symptomatic nodes. Therefore, as $\nu$ increases, we anticipate that nodes remain asymptomatic for a much shorter period of time and consequently the decay ratio will be small. For small values of $\nu$, nodes tend to remain asymptomatic for a longer period of time during which they continue to activate and connect at the same rate as a susceptible node; consequently, the epidemic sustains in the population. Figure \ref{fig:ex11} shows that our model captures the above phenomenon. \hfill \oprocendsymbol
\end{example}

\subsection{Evolution of the A-SIS epidemic and comparison with \cite{ogura2019optimal}}

We now numerically illustrate that the bound obtained in Theorem \ref{thm:asis_main} above better captures the evolution of the epidemic (fraction of infected population) compared to the bound on the decay ratio obtained in \cite{ogura2019optimal} which we denote by $\rho_p$. We consider two settings; one where $\rho_p$ is smaller than $\rho^*$ and second where $\rho_p$ is larger than $\rho^*$ as described in the following two examples, respectively.

\begin{figure*}[tb]
	\centering
	\includegraphics[scale=1]{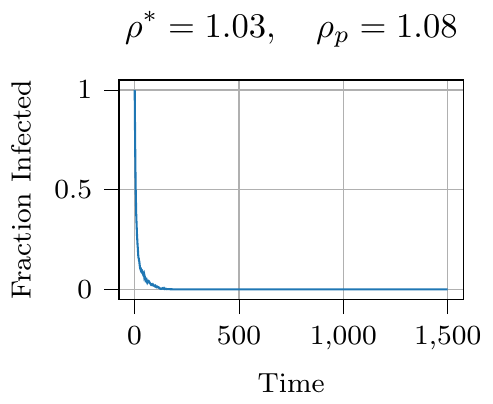}
	\includegraphics[scale=1]{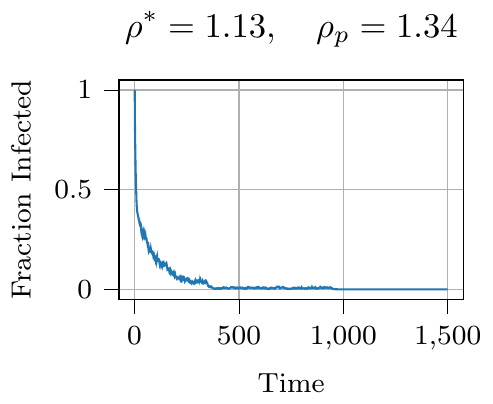}
	\includegraphics[scale=1]{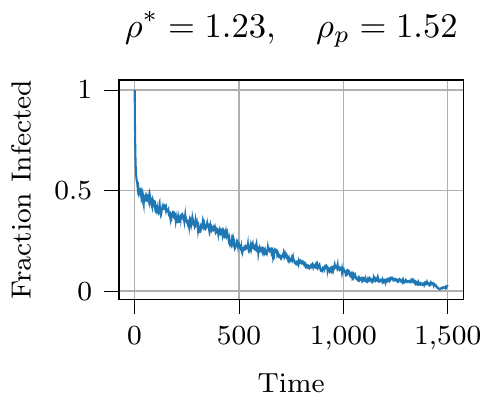} \\[2mm]
	\includegraphics[scale=1]{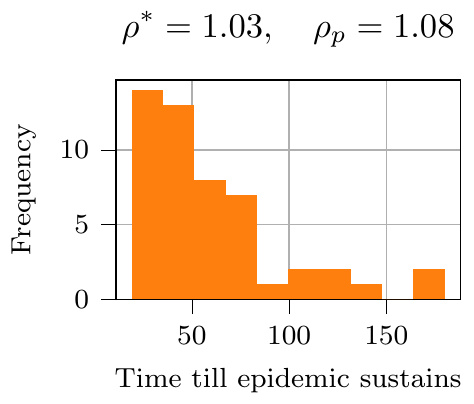}
	\includegraphics[scale=1]{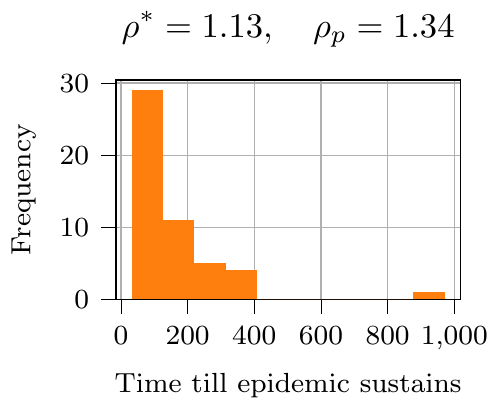}
	\includegraphics[scale=1]{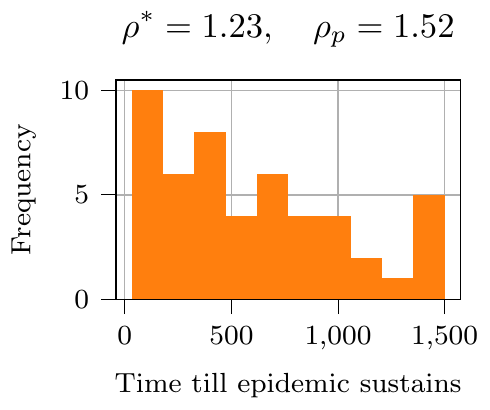}
	\caption{Evolution of fraction of infected nodes averaged across $50$ runs and histogram of time till the epidemic persists with parameters described in Example \ref{ex:2}.}
	\label{fig:ex2}
\end{figure*}

\begin{example}\label{ex:1}
We consider a set of $n=50$ nodes and set the infection rate $\beta = 0.2$, recovery rate $\delta = 0.15$ and $m=5$. We assume that $40$ out of $50$ nodes have activity rate $a_i = 0.1$, adaptation factor $\chi_i = 0.05$ and acceptance rate $\pi_i = 0.98$. For the remaining $10$ nodes, we choose $a_i = 0.6$ and $\pi_i = 0.03$ and vary the adaptation parameter which results in varying values of $\rho^*$ and $\rho_p$. We initialize with $40$ nodes being infected and $10$ nodes being susceptible and simulate the A-SIS epidemic. The average fraction of infected nodes across the $50$ independent runs for a duration of $1500$ time steps is reported in Figure \ref{fig:ex1}. The bounds $\rho^*$ and $\rho_p$ are shown in the titles of the plots. If at a given point of time (before the maximum time-step 1500), all nodes are susceptible (i.e., the underlying Markov chain has reached the disease-free absorbing state) then the simulation ends. 

The plot on the top left panel of Figure \ref{fig:ex1} corresponds to the case with $\chi_i = 0.95$ for the second group of $10$ nodes and shows that the epidemic sustains in the population. For other cases, the bounds are smaller and it results in the states reaching the disease-free absorbing state of the dynamics. We also plot the histogram of the time the simulation ends for $\chi_i = 0.7$ and $\chi_i = 0.4$ (for the second group of nodes) over the $50$ runs in the right panel of Figure \ref{fig:ex1}. We note that for these two cases, the epidemic sustains in the population despite the upper bound on the decay ratio obtained in \cite{ogura2019optimal} being smaller than $1$. In contrast, when $\rho^* < 1$ (bottom left panel), the epidemic reaches the absorbing state in less than $150$ iterations in all the $50$ independent runs. \hfill \oprocendsymbol
\end{example}

The above example shows that the epidemic sustains in the population even when $\rho_p < 1$ (but $\rho^* > 1$). In the following example, we consider parameters with $\rho^* < \rho_p$ and show that the epidemic reaches the disease-free state much faster when $\rho^*$ is close to $1$ even when $\rho_p$ is relatively large. 

\begin{example}\label{ex:2}
We consider a similar setting as above with a set of $n=50$ nodes and set the rate of infection $\beta = 0.4$, the rate of recovery $\delta = 0.15$ and $m=5$. We assume that $40$ out of $50$ nodes have activity rate $a_i = 0.1$, adaptation factor $\chi_i = 0.05$ and acceptance rate $\pi_i = 0.02$. For the remaining $10$ nodes, we choose $a_i = 0.6$ and $\pi_i = 0.75$ and vary the adaptation parameter (three values with $\chi_i= 0.05, 0.5$ and $0.95$) which results in varying values of $\rho^*$ and $\rho_p$.

We initialize with all nodes being infected and simulate the A-SIS epidemic $50$ times. The average fraction of infected nodes across the $50$ independent runs for a duration of $1500$ time steps and the histogram of the time the simulation ends are shown in the top and bottom panel of Figure \ref{fig:ex2}, respectively. As the bounds increase, the epidemic sustains for a longer time period. However, despite a relatively large value of $\rho_p$ (but with $\rho^*$ closer to $1$), the epidemic does not sustain for the entire duration in all the simulations. \hfill \oprocendsymbol
\end{example}

A stark contrast in results can be observed in the top left panel of Figure \ref{fig:ex1} and the top right panel of Figure \ref{fig:ex2}; in the former, the epidemic sustains for $\rho_p = 1.005$ while in the latter it reaches the disease-free state in most runs even when $\rho_p = 1.52$. To summarize, in both the examples considered above, the bound $\rho^*$ derived in our work better captures the evolution of the A-SIS epidemic. 

\section{Discussion and Conclusion}\label{sec:discussion}

In this paper, we propose a new activity-driven adaptive epidemic model that includes asymptomatic carriers present in several infectious diseases. In the proposed model, symptomatic individuals reduce their activation and acceptance probabilities while asymptomatic individuals do not, potentially because they are not aware of being infected. We show that the proposed model captures several existing epidemic models as special cases. We derive a linearized dynamics that upper bounds the exact Markovian evolution by resorting to a mean-field approximation. We also highlight a potential inaccuracy in the upper bound on the decay ratio derived in \cite{ogura2019optimal} for the A-SIS epidemic model and generalize their results. The simulation results illustrate that the bound derived in our work better captures the evolution of the A-SIS epidemic compared to the bound obtained in \cite{ogura2019optimal}. 

Our work is an early attempt to develop an epidemic model with asymptomatic carriers on temporal networks. There are several promising avenues for future research in this context. The condition based on the decay ratio is only sufficient to guarantee that the disease is quickly eradicated from the population. In contrast, in the classical SIS epidemic model, when the decay ratio is larger than $1$, there exists a unique endemic state that serves as an equilibrium of the mean-field dynamics. While we conjecture that the A-SIS and A-SIYS epidemic models would have a similar behavior, an analogous result has not yet been formally established. Similarly, developing scalable centralized and decentralized protection schemes for containing the epidemic in large-scale networks in presence of asymptomatic carriers is yet another challenging open problem.

\bibliographystyle{IEEEtran}
\bibliography{refs,refs_new}

\end{document}